\documentclass[12pt]{article}
\usepackage{amsmath,amsthm,amssymb,bbm,txfonts}
\usepackage{stmaryrd}
\usepackage{esint}
\usepackage[round]{natbib}

\newtheorem{theorem}{Theorem}[section]
\newtheorem{Def}{Definition}[section]
\newtheorem{prop}{Proposition}[section]
\newtheorem{lemma}{Lemma}[section]

\newtheorem{remark}{Remark}[section]
\newtheorem{example}{Example}[section]
\newtheorem{assump}{Assumption}[section]

\newcommand{\E}{\mathbb{E}}
\newcommand{\pr}{\mathbb{P}}
\newcommand{\R}{\mathbb{R}}

\newcommand{\norm}[1]{\left|\left|#1\right|\right|}
\newcommand{\parenth}[1]{\left(#1\right)}
\newcommand{\expect}[1]{\mathbb E\left[#1\right]}
\newcommand{\integral}[3]{\int_{#1}^{#2}{#3}\;}
\newcommand{\set}[2]{\{{#1:\;}{#2}\}}

\title{It\^o\rq{}s formula for finite variation L\'evy processes: The case of non-smooth functions}

\author{Ramin Okhrati\thanks{University of Southampton, Southampton, UK, Email: r.okhrati@soton.ac.uk }
~, Uwe Schmock\thanks{Vienna University of Technology, Vienna, Austria, Email: schmock@fam.tuwien.ac.at}
}

\date{}
\makeatletter

\newcommand{\Rmnum}[1]{\expandafter\@slowromancap\romannumeral #1@}

\usepackage{fullpage}
\begin{document}
\maketitle
\begin{abstract}
Extending It\^o's formula to non-smooth functions is important both in theory and applications. One of the fairly general extensions of the formula, known as Meyer-It\^o, applies to one dimensional semimartingales and convex functions. There are also satisfactory generalizations of It\^o's formula for diffusion
processes where the Meyer-It\^o assumptions are weakened even further. 
We study a version of It\^o\rq{}s formula for multi-dimensional finite variation L\'evy processes assuming that the underlying function is continuous and admits weak derivatives. We also discuss some applications of this extension, particularly in finance.

\end{abstract}
\thispagestyle{empty}
\clearpage

\pagenumbering{arabic}

\setcounter{equation}{0}
\setlength{\baselineskip}{1.3\baselineskip}

\section{Introduction}

In order to motivate our study, we consider the following Partial Integro-Differential Equation (PIDE):
\begin{align}\label{eq:pide}
    \frac{\partial P}{\partial t}(t,x)&+rx\frac{\partial P}{\partial x}(t,x)+\frac{\sigma^2x^2}{2}\frac{\partial^2 P}{\partial x^2}(t,x)-rP(t,x)\nonumber\\
    &+\int v(dy)\left(P(t,xe^y)-P(t,x)-x(e^y-1)\frac{\partial P}{\partial x}(t,x)\right)=0,\nonumber\\
   &  \hspace{-1.25cm}P(T,x)=(x-K)^+ \text{, for all $x\in(0,D)$,}\nonumber\\
  &\hspace{-1.25cm} P(t,x)=0 \text{, for all $x\geq D$, and all $t\in[0,T],$}
\end{align}
where $D>K>0$, $r>0$, $T>0$, are constants, and $v$ is the L\'evy measure of a L\'evy process $X$ with characteristic triplet $(\sigma^2,v,\gamma)$. Furthermore, it is assumed that $\left(e^{X_t}\right)_{0\leq t\leq T}$ is a martingale with respect to the natural filtration generated by $X$ and a risk-neutral probability measure. 

Finding the solution of this PIDE (or similar ones) is of particular interest in different applied fields. For instance, under some circumstances the solution of PIDE \eqref{eq:pide} can be identified as the price of a financial security. As it follows, It\^o's formula is a key element in this procedure.

More precisely, assume that the risk-neutral evolution of an asset is modeled by $S_t=S_0e^{rt+X_t}$, where $r$ and $X$ are the same as above such that $\left(e^{-rt}S_t\right)_{0\leq t\leq T}$ is a martingale under the risk-neutral probability measure.  Suppose that we are interested in pricing a barrier option with maturity $T$, strike price $K$, barrier $D>K$, and the payoff $\max(S_T-K,0)1_{\{\max_{0\leq t\leq T}S_t<D\}}$.   If $\sigma>0$, then using It\^o\rq{}s formula one can show that there is a $C^{1,2}$ solution of  PIDE \eqref{eq:pide} which is in fact the price of this barrier option given by 
\begin{equation}\label{eq:purposed}
     P(t,x)=e^{-r(T-t)}\mathbb E[H(S_{T\wedge\tau_D})|S_t=x],
\end{equation}
where $\E$ is the expectation under the risk-neutral measure, $H(x):=(x-K)^+1_{\{x<D\}}$, and $\tau_D:=\inf\{s\geq t; X_s\geq D\}$, see Proposition 12.2 of \cite{cont}.

Equation \eqref{eq:purposed} is in fact the Feynman-Kac representation of the solution of PIDE \eqref{eq:pide} which can be numerically calculated through simulation techniques.  Note that the condition $\sigma>0$ is crucial for this argument to work which guarantees that the purposed solution \eqref{eq:purposed} is smooth and hence It\^o's formula is applicable.  However, in the case of pure jump L\'evy processes, i.e. when $\sigma=0$, the smoothness is not obvious and it can fail. The situation is more complicated for American options where the smoothness of the purposed solution is not known even in the presence of a non-zero volatility, see Chapter 12 of \cite{cont} for more detail. For example, Theorem 7.2 of \cite{boy} shows that the smoothness of the purposed solution in the case of American option fails for tempered stable L\'evy processes with finite variation.

One purpose of this work is to fix this kind of problems for models using finite variation L\'evy processes. For this class of processes, under some conditions, we obtain an It\^o formula that works well with non-smooth continuous functions. In particular, this can provide a solution to PIDE \eqref{eq:pide} when $\sigma=0$ and $X$ is a finite variation L\'evy process. This problem is investigated at the end of this paper. We continue with some literature review.

A version of It\^o's formula is obtained in \cite{aebi}  where the underlying process is a continuous semimartingale with a special structure. In this paper, the first and second order derivatives of the function are defined in the sense of distributions and they satisfy some local integrability conditions. \cite{follmer2} discuss an extension of the formula to a one-dimensional standard Brownian motion and an absolutely continuous function with a locally square integrable derivative. This result was further extended by \cite{follmer1} to a multi-dimensional Brownian motion. 

Following the idea of \cite{follmer2}, an extension of It\^o's formula is proved in \cite{bardina1} for a one-dimensional diffusion process such that its law has a density satisfying certain integrability conditions.  In their work, it is assumed that the underlying function $f=f(t,x)$ is absolutely continuous in $x$ with a locally square integrable derivative satisfying a mild form of continuity in time $t$.

In all the above works, the sample paths of the underlying processes are continuous. Concerning discontinuous processes,  Theorem 70, Chapter \Rmnum{4} of \cite{protter} (known as Meyer-It\^o's formula) provides a fairly general extension of It\^o's formula to semimartingales and one dimensional  convex functions. 

Comparing to Theorem 70, Chapter \Rmnum{4} of \cite{protter}, our extension applies to finite variation L\'evy processes and continuous functions that admit weak derivatives. Therefore this generalizes Meyer-It\^o's formula  for finite variation L\'evy processes. In addition, it is assumed that the function is multi-dimensional and time-dependent.  Beside the motivation provided at the beginning and theoretical interests to extend It\^o's formula for these processes, it is also argued in \cite{geman} that the evolution of asset prices are better modeled by finite variation processes with infinite activity\footnote{A L\'evy process $X$ in $\R^d$ is of infinite activity, if there are infinite number of jumps on any finite time interval, i.e. $v(\R^d)=\infty$, where $v$ is the L\'evy measure of $X$. }. 

The structure of the paper is as follows. The theoretical backgrounds, in particular some fundamental results in real and functional analysis are reviewed in Section \ref{sec:pd}. Section \ref{sec:dakt} concentrates on hypotheses and key tools. The main result is proved in Section \ref{sec:mr}. Finally, the paper ends with some examples and conclusions.


\section{Preliminaries and Definitions}\label{sec:pd}
In this section, we recall a few results from real and functional analysis (basically  Distribution theory) that will be used later. We begin with some definitions. In what follows, $\R$ is the set of real numbers; $U\subset\R^d$ is a nonempty open set, $d\geq 1$; $|.|$ and $\norm{.}_d$ are respectively the one-dimensional and d-dimensional Euclidean norms; and $m$ is the Lebesgue measure. For simplicity, regardless of the dimension of the space, the Lebesgue measure is always denoted by $m$.
\begin{Def}
   A point $x\in U\subset\R^d$ is a Lebesgue point of a function $f:U\longmapsto\R$ if $$\lim_{r\rightarrow0^+}\frac{1}{m(B_r(x))}\integral{B_r(x)}{}{|f(y)-f(x)|}dy=0,$$ where $B_r(x)=\set{y\in\R^d}{|y-x|<r}$ and the limit is taken for those $r$ small enough to guarantee that $B_r(x)$ is a subset of $U$.
\end{Def}
\begin{Def}
   The set of all Lebesgue points of $f:U\longmapsto\R$ is denoted by $L_f$ and it is called the Lebesgue set.
\end{Def}
\begin{Def}
   A family $\{E_r\}_{r>O}$ of Borel subsets of $U$ is said to shrink nicely to $x\in U$ if the following two conditions hold
   \begin{itemize}
      \item $E_r\subset B_r(x)\subset U$ for each $r$,
      \item there is a constant $\alpha > 0$, independent of $r$, such that $m(E_r)> \alpha m(B_r(x))$.
   \end{itemize}
\end{Def}
\begin{theorem}\label{theorem:ldt}
   \textbf{The Lebesgue Differentiation Theorem.} Suppose that $f\in L_{loc}^1(U)$ and $supp(f)\subset U$. Then we have
   \begin{itemize}
      \item $m(U-L_f)=0$,
      \item for every $x$ in the Lebesgue set of $f$, in particular for almost every $x$ in $U$, we have
      $$\lim_{r\rightarrow0^+}\frac{1}{m(E_r)}\integral{E_r}{}{|f(y)-f(x)|}dy=0,$$ where $\{E_r\}_{r>0}$ is a family of Borel subsets of $U\subset\R^d$ that
      shrinks nicely to $x$.
   \end{itemize}
\end{theorem}
For a proof of this theorem in the case of $U=\R^d$, see Theorem 3.21 of \cite{folland}. The generalization to an open set $U\subset\R^d$ is straightforward. Note that following the Lebesgue Differentiation Theorem we have $\lim_{r\rightarrow0^+}\frac{1}{m(E_r)}\integral{E_r}{}{f(y)}dy=f(x),$
where $f$ and $E_r$ are the same as in the above theorem. Therefore, this can be thought of as a generalization of the fundamental theorem of calculus. In general, determining the Lebesgue points of a function is not an easy task. The next lemma gives a partial answer to this challenge; the proof is simple and hence omitted.

\begin{lemma}\label{lemma:fcont}
   If $f\in L_{loc}^1(U)$, $U\subset\R^d$, and $f$ is continuous at $x\in U$, then $x\in L_f.$
\end{lemma}
\begin{Def}\label{def:convolution}
     If $g:\R^d\longmapsto\R$ and $f:\R^d\longmapsto\R$, are measurable functions, then the convolution $g*f:\R^d\longmapsto\R$ is defined by $(g*f)(x):=\integral{\R^d}{}{g(x-y)f(y)}dy,$ provided that for every $x$ in $\R^d$, the integral is well defined.
\end{Def}
Some basic properties of convolution can be found in standard text books such as \cite{folland} or \cite{brezis}. The next lemma provides a simple and sufficient condition for the existence of convolution.
\begin{lemma}\label{lemma:well defined}
      Let $f\in L_{loc}^p(U)$, $p\geq1$, and $supp(f)\subset U$. Suppose that $g:\R^d\longmapsto\R$ is bounded and compactly supported. Then $g*f1_U$, $f1_U(x)=\left\{
         \begin{array}{ll}
           f(x), & \hbox{$x\in U;$} \\
           0, & \hbox{$x\notin U,$}
         \end{array}
       \right.$ is well defined on $\R^d$, i.e. the integral $\integral{U}{}{g(x-y)f(y)}dy$ is finite for all $x$ in $\R^d$.
\end{lemma}

Let $\eta$ be any function in $C_c^{\infty}(\R^d)$ such that it satisfies the following conditions
$$\eta\geq0,\qquad \integral{\R^d}{}{\eta(x)}dx=1,\qquad supp(\eta)=\overline{B_1(0)}.$$
For any $\epsilon>0$, define $\eta^{\epsilon}(x)=\frac{1}{\epsilon^d}\eta(\frac{x}{\epsilon})$ then clearly we have
$$\eta^{\epsilon}\in C_c^{\infty}(\R^d),\qquad \integral{\R^d}{}{\eta^{\epsilon}(x)}dx=1,\qquad supp(\eta^{\epsilon})=\overline{B_{\epsilon}(0)}.$$
The next definition provides an example of such a function.
\begin{Def}
   Let $$\eta(x)=\left\{
         \begin{array}{ll}
           ce^{\frac{-1}{1-\norm{x}_d^2}}, & \hbox{$\norm{x}_d<1;$} \\
           0, & \hbox{$\norm{x}_d\geq1,$}
         \end{array}
       \right.$$
   and take $c$ such that $\integral{\R^d}{}{\eta(x)}dx=1$. Then $\eta^\epsilon$ is called the standard mollifier.
\end{Def}
Our discussion does not depend on a specific choice of $\eta^\epsilon$. However, if necessary, the reader can always consider the standard mollifier. Suppose that $f\in L_{loc}^p(U)$, $p\geq1$, and for every $\epsilon>0$, let $f^\epsilon:\R^d\longmapsto\R$ be defined by $$f^\epsilon(x):=(\eta^\epsilon*f1_U)(x)=\integral{U}{}{\eta^\epsilon(x-y)f(y)}dy.$$
For a fixed $x$ and $\epsilon$ small enough (that depends on $x$), $\overline{B_\epsilon(x)}\subset U$ and so $f^\epsilon(x)$ exists. However, if $supp(f)\subset U$ and since $f\in L_{loc}^p( U)$, $p\geq1$, by Lemma \ref{lemma:well defined}, $f^\epsilon$ is well defined on $\R^d$ for all $\epsilon>0$. The following theorem is a classical well-known result in the theory of distributions. Parts (1) and (2) can be found in Section 4.4 of \cite{brezis}, and part (3) is a  conclusion of Theorem \ref{theorem:ldt}.
\begin{theorem}\label{theorem:convsmooth}
   Assume that $f\in L_{loc}^p(U)$, $p\geq1$, $supp(f)\subset U$, and $\epsilon>0$. Then
   \begin{enumerate}
      \item $f^\epsilon\in C^{\infty}(\R^d)$ and $\partial^\alpha f^\epsilon=\partial^\alpha\eta^\epsilon*f1_U$,
      \item $f^\epsilon\longrightarrow f1_U$ in $L_{loc}^p(\R^d)$ as $\epsilon\rightarrow0^+$,
      \item $f^\epsilon\longrightarrow f1_U$ pointwise on $L_{f1_U}$ as $\epsilon\rightarrow0^+$, hence $f^\epsilon\longrightarrow f$ pointwise on $L_f$ as $\epsilon\rightarrow0^+$.
   \end{enumerate}
\end{theorem}
Note that part (3) of Theorem \ref{theorem:convsmooth} implies that $f^\epsilon\longrightarrow f$, Lebesgue almost every where on $U$.
Let $\mathbb N_0$ be the set of non-negative integers and $\mathbb N_0^d=\set{(\alpha_1,\alpha_2,...,\alpha_d)}{\alpha_i\in\mathbb N_0, i=1,2,...,d}.$ An element of the set $\mathbb N_0^d$ is called a multi-index. In our extended version of It\^o's formula instead of classical strong differentiability, we apply weak differentiability which is defined below.
\begin{Def}\label{def:weak1}
    Suppose that $\alpha\in\mathbb N_0^d$ is a multi-index. We say that a function $f\in L_{loc}^1(U)$, $U\subset\R^d$, is weakly differentiable; and also its $\alpha$th-weak derivative denoted by $\partial^\alpha f\in L_{loc}^1(U)$, if $$\integral{U}{}{(\partial^\alpha f(x))\phi(x)}dx=(-1)^{|\alpha|}\integral{U}{}{f(x)(\partial^\alpha\phi(x))}dx,\; \text{for all}\;\phi\in C_c^\infty(U),$$
where $|\alpha|=\sum_{i=1}^d\alpha_i$, and the functions $\phi\in C_c^\infty(U)$ are called  test functions. 
\end{Def}
By applying Theorem \ref{theorem:convsmooth} and simple properties of weak derivatives, we can get the following theorem.
\begin{theorem}\label{theorem:convweak}
     Let $f\in L_{loc}^1(U)$ and $supp(f)\subset U$. We further assume that $f$ admits the weak derivative $\partial^\alpha f\in L_{loc}^1(U)$, then:
\begin{enumerate}
     \item $f^\epsilon\in C^\infty(\R^d)$, and $\partial^\alpha(f^\epsilon)=\eta^\epsilon*(\partial^\alpha f)$ on $U$,
     \item $\partial^\alpha(f^\epsilon)\longrightarrow\partial^\alpha f$ in $L_{loc}^1(U)$ as $\epsilon\rightarrow0^+$,
     \item $\partial^\alpha (f^\epsilon)\longrightarrow \partial^\alpha f$ pointwise on $L_{\partial^\alpha f}$ as $\epsilon\rightarrow 0^+$.
\end{enumerate}
\end{theorem}
\begin{remark}\label{remark:part1}
     Note that part(1) of Theorems \ref{theorem:convsmooth} and \ref{theorem:convweak} still holds if we replace $\eta^\epsilon$ by a test function.
\end{remark}
Though it is very simple, the next lemma is a key point in our discussion.
\begin{lemma}\label{lemma:key}
     Assume that $f\in L_{loc}^1(U)$ has the weak derivative $\partial^\alpha f\in L_{loc}^1(U).$ Suppose that $\phi\in C_c^\infty(\R^d)$ is a test function with support of $K$ such that $\phi(x)\geq0$, for all $x\in\R^d$ and $\integral{\R^d}{}{\phi(x)}dx=1$. Then for every $x\in\R^d$ we have $$|\partial^\alpha(f*\phi)(x)|\leq\sup_{z\in U\cap\Lambda(x)}|\partial^\alpha f(z)|,$$ where $\Lambda(x)=\set{y\in\R^d}{x-y\in K}$.
\end{lemma}
\begin{proof}
     By using Remark \ref{remark:part1} we get $$\partial^\alpha(f*\phi)(x)=(\phi*1_U\partial^\alpha f)(x)=\integral{U}{}{\phi(x-y)\partial^\alpha f(y)}dy=\integral{U\cap\Lambda(x)}{}{\phi(x-y)\partial^\alpha f(y)}dy.$$
Using this equation and the following inequalities, we get the result
\begin{align*}
|\partial^\alpha (f*\phi)(x)| & \leq \sup_{z\in U\cap\Lambda(x)}|\partial^\alpha f(z)|\integral{U\cap\Lambda(x)}{}{\phi(x-y)}
dy\\  &\leq \sup_{z\in U\cap\Lambda(x)}|\partial^\alpha f(z)|\integral{\R^d}{}{\phi(x)}dx=\sup_{z\in U\cap\Lambda(x)}|\partial^\alpha f(z)|.
\end{align*}
\end{proof}
\begin{remark}
      Note that the value of the right-hand side of the inequality in Lemma \ref{lemma:key} can be infinity. 
\end{remark}

\section{Discussion of Assumptions and Key Tools}\label{sec:dakt}
In applying classical It\^o's formula on smooth functions $f:[0,\infty)\times U\longmapsto\R$, $U\subset\R^d$, the differentiability at $t=0$ is understood by being the right-hand side derivative. Note that since the Lebesgue measure of $\{0\}\times U$ is zero, the weak derivatives of $f$ can be defined similar to Definition \ref{def:weak1}.

Assume that $f:[0,\infty)\times U\longmapsto\R$ is a Lebesgue measurable function. In accordance with Definition\ref{def:weak1}, we say that $f\in L_{loc}^1([0,\infty)\times U)$ has weak derivatives $\partial^\alpha f\in L_{loc}^1([0,\infty)\times U)$ if
\begin{equation}\label{eq:weak2}
\integral{[0,\infty)\times U}{}{(\partial^\alpha f(x))\phi(x)}dx=(-1)^{|\alpha|}\integral{[0,\infty)\times U}{}{f(x)(\partial^\alpha\phi(x))}dx,\; \text{for all}\;\phi\in C_c^\infty([0,\infty)\times U).
\end{equation}
Note that since a test function $\phi$ is smooth, its derivatives at the origin are understood as the right-hand side ones. The results of Section \ref{sec:pd} are stated for open subsets of $\R^d.$ However, $[0,\infty)\times U$ is not an open set. So in our first step we fix this problem by introducing an extended version of $f$.

Suppose that the function $f:[0,\infty)\times U\longmapsto\R$ is continuous on $[0,\infty)\times U$. This function can be continuously extended to a new function $\tilde f:\R\times U\longmapsto\R$:
\begin{equation}\label{eq:weakex}
\tilde f(t,x)=\left\{
         \begin{array}{ll}
           f(t,x), & \hbox{$(t,x)\in [0,\infty)\times U$;} \\
           f(-t,x), & \hbox{$(t,x)\in (-\infty,0)\times U.$}
         \end{array}
       \right.
\end{equation}
Now in addition assume that $f\in L_{loc}^1([0,\infty)\times U)$ and it is weakly differentiable in the sense of equation \eqref{eq:weak2}. Then one can easily show that $\tilde f\in L_{loc}^1(\R\times U)$ and it is weakly differentiable on the open set $\R\times U$ in the sense of Definition \ref{def:weak1}. The weak derivatives of $\tilde f$ can be stated explicitly based on weak derivatives of $f$. For instance in the case of $d=1$, one can easily check that
$$\frac{\partial\tilde f}{\partial t}(t,x)=\left\{
         \begin{array}{ll}
           \frac{\partial f}{\partial t}(t,x), & \hbox{$(t,x)\in [0,\infty)\times U$;} \\
           -\frac{\partial f}{\partial t}(-t,x), & \hbox{$(t,x)\in (-\infty,0)\times U,$}
         \end{array}
       \right.$$
and
$$\frac{\partial\tilde f}{\partial x}(t,x)=\left\{
         \begin{array}{ll}
           \frac{\partial f}{\partial x}(t,x), & \hbox{$(t,x)\in [0,\infty)\times U$;} \\
           \frac{\partial f}{\partial x}(-t,x), & \hbox{$(t,x)\in (-\infty,0)\times U,$}
         \end{array}
       \right.$$
where $\frac{\partial f}{\partial t}(t,x)$ and $\frac{\partial f}{\partial x}(t,x)$ are weak derivatives of $f$ in the sense of equation  \eqref{eq:weak2}.


Assume that $(\Omega,\mathfrak F,\pr)$ is a complete probability space. Let $X=(X_t)_{t\geq0}$, $X_t:\Omega\longmapsto U$, $U\subset\R^d$, be a c\`adl\`ag stochastic process that is defined on this space. In any extension of It\^o's formula, it is important to somehow measure the amount of time that the process spends in some certain regions of the domain. In particular, this is crucial for those points for which the function is not smooth. For instance, in the case of Meyer-It\^o formula (see Theorem 70, Chapter \Rmnum{4} of \cite{protter}), this is done through local times. In the next proposition, we discuss a similar tool which is a key result in our extension. The proposition is provided for a certain class of processes explained below.

\begin{assump}\label{assump:main}
     Suppose that $X:[0,\infty)\times\Omega\longmapsto U$ is a c\`adl\`ag stochastic process defined on the complete probability space $(\Omega,\mathfrak F,\pr)$, that satisfies the following condition: If $A\subset U$ is a Borel measurable set such that $m(A)=0$, where $m$ is the Lebesgue measure, then for all $s\in\R^+$, $\pr(X_s\in A)=0$. In other words, for all $s\in\R^+$, the measure $\mu_s$ on $U$ defined by $\mu_s(A)=\pr(X_s\in A)$, is absolutely continuous with respect to the Lebesgue measure.
\end{assump}
\begin{prop}\label{prop:main}
     Assume that the process $X$ satisfies Assumption \ref{assump:main}. Let $A\subset [0,\infty)\times U$ be any Lebesgue measurable set such that $m(A)=0$, then for all $t\geq0$ we have $$\pr\set{\omega\in\Omega}{m(\set{s\in[0,t]}{(s,X_s(\omega))\in A})=0}=1.$$
    In particular, this implicitly implies that for almost all $\omega\in\Omega$, the set $\set{s\in[0,t]}{(s,X_s(\omega))\in A}$ is Lebesgue measurable for all $t\geq0$.
\end{prop}
\begin{proof}
     First assume that $A$ is a Borel measurable set. Define the process $Y:[0,\infty)\times\Omega\longmapsto[0,\infty)\times U$ by $Y(s,\omega)=(s,X_s(\omega))$. The process $Y$ is c\`adl\`ag and by Proposition 1.21 of \cite{jacod}, $Y$ is $\mathfrak B_{[0,\infty)}\times\mathfrak F$ measurable, where $\mathfrak B_{[0,\infty)}$ is the Borel $\sigma$-algebra on $[0,\infty)$ and $\mathfrak F$ is the $\sigma$-algebra on $\Omega$. Hence $Y^{-1}(A)$ belongs to $\mathfrak B_{[0,\infty)}\times\mathfrak F$ and so $\llbracket0,t\rrbracket\cap Y^{-1}(A)$ is in $\mathfrak B_{[0,\infty)}\times\mathfrak F\subset\mathcal L\times\mathfrak F$, where $\llbracket0,t\rrbracket=[0,t]\times\Omega$, and $\mathcal L$ is Lebesgue $\sigma$-algebra on $[0,\infty)$. Therefore the function $f:[0,\infty)\times\Omega\longmapsto\R$ defined by $f:=1_{\llbracket0,t\rrbracket\cap Y^{-1}(A)}$ belongs to $L^1(m\times\pr)$.

From Fubini-Tonelli Theorem, see Theorem 2.37 of \cite{folland}, it follows that $f_\omega$ defined by $f_\omega:=(.,\omega)$ is in $ L^1(m)$ for almost all $\omega$. So for a fixed $\omega$, $\llbracket0,t\rrbracket\cap Y^{-1}(A)$ is Lebesgue measurable, and $m\set{s\in[0,t]}{(s,X_s(\omega))\in A}$ is well defined for almost all $\omega\in\Omega$.

Moreover, let $Z(\omega):=\integral{}{}{f_\omega}dm=m\set{s\in[0,t]}{\left(s,X_s(\omega)\right)\in A}$, then again by Fubini-Tonelli Theorem $Z$ is a random variable and $Z\in L^1(\pr)$, furthermore,  we can calculate its expectation
\begin{align*}
\expect{Z} &= \int\int f_\omega\;dm\;d\pr=\integral{0}{t}{\integral{}{}{f_s}d\pr}ds
\\  &= \integral{0}{t}{\expect{1_{\{(s,X_s)\in A\}}}}ds.
\end{align*}
Note that for a  fixed $s$, $1_{\{(s,X_s)\in A\}}=1_{\{X_s\in A_s\}}$, where $A_s=\set{y\in\R^d}{(s,y)\in A}$ is Borel measurable, hence we obtain
\begin{equation}\label{eq:mainprop1}
     \expect{Z}=\integral{0}{t}{\pr(X_s\in A_s)}ds.
\end{equation}
The set $A$ is Borel measurable and hence Lebesgue measurable as well. By Theorem 2.36 of \cite{folland} the function $s\longmapsto m(A_s)$ is Lebesgue measurable and $m(A)=\integral{[0,\infty)}{}{m(A_s)}ds$. By the proposition's assumption $m(A)=0$ which concludes that $m(A_s)=0$ for Lebesgue almost all $s\geq0$, i.e. there exists a set $N\subset[0,\infty)$ such that $m(N)=0$ and if $s\notin N$ then $m(A_s)=0.$ From equation \eqref{eq:mainprop1} and Assumption \ref{assump:main}, we get
$$\expect{Z}=\integral{[0,t]\cap N^c}{}{\pr(X_s\in A_s)}ds=\integral{[0,t]\cap\set{s}{m(A_s)=0}}{}{\pr(X_s\in A_s)}ds=0.$$
The random variable $Z$ is non-negative and $\expect{Z}=0$, hence $Z=0$, $\pr$-almost surely which means that for almost all $\omega\in\Omega$, $m(\set{s\in[0,t]}{(s,X_s(\omega))\in A})=0$. This completes the proof when $A$ is Borel measurable.

Next, suppose that $A$ is a Lebesgue measurable set, then $A=A^{'}\cup A^{''}$, $A^{''}\subset B$, where $A^{'}$ and $B$ are Borel measurable and $m(B)=0$. Now the result follows from the previous part, and the facts that $m(A)=0$ and the probability space is complete.
\end{proof}
Note that if $A=[0,t]\times B$, where $B\subset\R^d$ a Borel set, then $\set{s\in[0,t]}{(s,X_s)\in A}$ is the amount of time that the process $X$ spends in Borel set $B$. So under Assumption \ref{assump:main}, Proposition \ref{prop:main} concludes that almost surely the Lebesgue measure of this amount of time is zero for any zero Borel measurable set. 

We would like to point out that this measure can be quite different than local times. For instance, let $X$ be a standard Brownian motion, then by Proposition \ref{prop:main}, $m\set{s\in[0,t]}{X_s=a}=0$, $\pr$-almost surely for all real numbers $a$ whereas the local time of a Brownian motion at the level $a$ is not zero. This is also because of the fact that as a measure the local time of a Brownian motion is singular with respect to the Lebesgue measure.
\section{The Main Result}\label{sec:mr}
In this section, we state and prove our main result. First, we mention that the result holds for a finite variation L\'evy process that satisfies Assumption \ref{assump:main}. This assumption is not valid for a compound Poisson process $X$ as $\pr[X_t=0]>0$, for $t>0$, and therefore, the measure $\mu_t$ defined in Assumption \ref{assump:main}  is not absolutely continuous with respect to the Lebesgue measure, see Remark 27.3 of \cite{sato}. 
However, based on Theorem 27.7 of \cite{sato}, Assumption \ref{assump:main} is always satisfied for a finite variation L\'evy process with infinite activity, if its L\'evy measure is absolutely continuous with respect to the Lebesgue measure.

 For simplicity we present the theorem for the case of $d=1$, however there is no restriction on extending the result to a general $d$.
\begin{theorem}\label{theorem:main}
     Assume that $f:[0,\infty)\times U\longmapsto\R$ is a continuous function on $[0,\infty)\times U$ such that $f\in L_{loc}^1([0,\infty)\times U)$, $supp(f)\subset[0,\infty)\times U$, and $U$ is an open set of $\R$. Let the weak derivatives $\frac{\partial f}{\partial s},\frac{\partial f}{\partial x}\in L_{loc}^1([0,\infty)\times U)$ be locally bounded and defined by equation \eqref{eq:weak2}. Suppose that $X$ is a finite variation L\'evy process satisfying Assumption \ref{assump:main}  such that for all $t\geq0$, $X_t$ and $X_{t^-}$ are in $U$. Then
\begin{align*}
f(t,X_t) = f(0,X_0)&+\integral{0}{t}{\frac{\partial f}{\partial s}(s,X_s)}ds+\gamma\integral{0}{t}{\frac{\partial f}{\partial x}(s,X_s)}ds
\\  &+ \iint\limits_{[0,t]\times \R}\parenth{f(s,X_{s^-}+x)-f(s,X_{s^-})}\;J_X(ds\times dx),
\end{align*}
where $J_X$ and $\gamma$ are respectively the Poisson random measure and the drift coefficient of the process $X$ admitting the following representation: $X_t=\gamma t+\integral{[0,t]\times \R}{}{x}J_X(ds\times dx)$.
\end{theorem}
\begin{proof}
     Assume that $\tilde f$ is an extension of the function $f$ to $\R\times U$ given by equation \eqref{eq:weakex}, note that $supp(\tilde f)\subset\R\times U$. Let $\phi_n=\eta^{\frac{1}{n}}$ and $f_n(t,x):=(\phi_n*\tilde f1_{\R\times U})(t,x)$, where $(t,x)\in\R^2$, $n\geq1$, and $\eta^{\frac{1}{n}}$ is defined in Section \ref{sec:pd}. Since $\tilde f\in L_{loc}^1(\R\times U)$, by Theorem \ref{theorem:convweak}, $f_n\in C^\infty(\R\times\R)$ for all $n\geq1$. Hence from It\^o's formula, see Theorem 4.2 of \cite{kyprianou}, we have
\begin{align*}
f_n(t,X_t) = f_n(0,X_0)&+\integral{0}{t}{\frac{\partial f_n}{\partial s}(s,X_s)}ds+\gamma\integral{0}{t}{\frac{\partial f_n}{\partial x}}ds
\\  &+ \iint\limits_{[0,t]\times \R}\parenth{f_n(s,X_{s^-}+x)-f_n(s,X_{s^-})}\;J_X(ds\times dx).
\end{align*}
The rest of the proof is divided into five steps:

\textbf{Step 1.} Since $\tilde f$ is a continuous function, by Lemma \ref{lemma:fcont} $L_{\tilde f}=\R\times U$. On the other hand for all $t\geq0$, $X_t$ is in $U$ and so by Theorem \ref{theorem:convsmooth}, $f_n(t,X_t)\longrightarrow\tilde f(t,X_t)$, for all $\omega\in\Omega$ and $t\in\R$. Especially
$f_n(0,X_0)\longrightarrow\tilde f(0,X_0)$. Also note that for $t\geq0$, $\tilde f(t,X_t)=f(t,X_t)$ by the definition of $\tilde f$.

\textbf{Step 2.} From Theorem \ref{theorem:convweak},  if $(s,X_s)\in L_{\frac{\partial\tilde f}{\partial s}}$, then we have $$\frac{\partial f_n}{\partial s}(s,X_s)\longrightarrow\frac{\partial\tilde f}{\partial s}(s,X_s).$$ Let $L_1=\R\times  U-L_{\frac{\partial\tilde f}{\partial s}}$, then
\begin{align*}
      \integral{0}{t}{\frac{\partial f_n}{\partial s}(s,X_s)}ds &= \integral{0}{t}{\frac{\partial f_n}{\partial s}(s,X_s)1_{\{(s,X_s)\notin L_1\}}}ds+\integral{0}{t}{\frac{\partial f_n}{\partial s}(s,X_s)1_{\{(s,X_s)\in L_1\}}}ds.
\end{align*}
By Theorem \ref{theorem:ldt}, $m(L_1)=0$, therefore by Proposition \ref{prop:main}, $m\set{s\in[0,t]}{(s,X_s)\in L_1}=0$, $\pr$-almost surely. Hence because of the properties of Lebesgue integral, for each fixed $t$, the integral $$\integral{0}{t}{\frac{\partial f_n}{\partial s}(s,X_s)1_{\{(s,X_s)\in L_1\}}}ds=\integral{[0,t]\cap\{s:\;(s,X_s)\in L_1\}}{}{\frac{\partial f_n}{\partial s}(s,X_s)}ds,$$ is zero $\pr$-almost surely. Therefore for a fixed $t$, $$\integral{0}{t}{\frac{\partial f_n}{\partial s}(s,X_s)}ds=\integral{0}{t}{\frac{\partial f_n}{\partial s}(s,X_s)1_{\{(s,X_s)\notin L_1\}}}ds,\;\pr-\;\text{almost surely}.$$  
By Lemma \ref{lemma:key}, for all $(s,x)\in\R^2$, $|\frac{\partial f_n}{\partial s}(s,x)|\leq\sup_{z\in (\R\times U)\cap\Lambda(s,x)}|\frac{\partial\tilde f}{\partial s}(z)|\leq \sup_{z\in \Lambda(s,x)}|\frac{\partial\tilde f}{\partial s}(z)|$, where $\Lambda(s,x)=\set{y\in\R^2}{(s,x)-y\in K}$, and $K=\sup\phi_n=\overline{B_{\frac{1}{n}}(0)}\subset\overline{B_1(0)}$ which results $$|\frac{\partial f_n}{\partial s}(s,X_s)|\leq\sup_{z\in\Lambda(s,X_s)}|\frac{\partial\tilde f}{\partial s}(z)|,\;0\leq s\leq t.$$

For a fixed $\omega\in\Omega$, $\Lambda(s,X_s)$ is bounded, because $X$ is bounded on $[0,t]$ (due to being a c\`adl\`ag process). Therefore for a fixed $\omega\in\Omega$ and $s\in[0,t]$, one can find an upper bound for $|\frac{\partial f_n}{\partial s}(s,X_s)|$ that depends only on $\omega$, $t$, and the minimum, maximum of $\frac{\partial\tilde f}{\partial s}(s,X_s)$ on $[0,t]$. This upper bound is finite because the weak derivatives of $f$ are locally bounded by the assumption of the theorem and so the weak derivatives of $\tilde f$ must be locally bounded too. Therefore, one can apply Lebesgue Dominated Convergence theorem and we obtain:
$$\lim_{n\rightarrow\infty}\integral{0}{t}{\frac{\partial f_n}{\partial s}(s,X_s)}ds=\integral{0}{t}{\lim_{n\rightarrow\infty}\frac{\partial f_n}{\partial s}(s,X_s)1_{\{(s,X_s)\notin L_1\}}}ds,\;\pr-\;\text{almost surely}.$$ By Theorem \ref{theorem:convweak}, this is $\pr$-almost surely equal to $\integral{0}{t}{\frac{\partial\tilde f}{\partial s}(s,X_s)1_{\{(s,X_s)\notin L_1\}}}ds$. Since $\pr$-almost surely, $m\set{s\in[0,t]}{(s,X_s)\in L_1}=0$, and for each $s\in[0,t]$, $\frac{\partial\tilde f}{\partial s}(s,X_s)=\frac{\partial f}{\partial s}(s,X_s)$, we have
$$\lim_{n\rightarrow\infty}\integral{0}{t}{\frac{\partial f_n}{\partial s}(s,X_s)}ds=\integral{0}{t}{\frac{\partial f}{\partial s}(s,X_s)}ds,\;\pr-\;\text{almost surely}.$$

\textbf{Step 3.} Similar to Step 2, one can prove that $$\lim_{n\rightarrow\infty}\integral{0}{t}{\frac{\partial f_n}{\partial x}(s,X_s)}ds=\integral{0}{t}{\frac{\partial f}{\partial x}(s,X_s)}ds,\;\pr-\;\text{almost surely}.$$

\textbf{Step 4.} Let $I_n=\iint\limits_{[0,t]\times\R}\parenth{f_n(s,X_{s^-}+x)-f_n(s,X_{s^-})}\;J_X(ds\times dx)$, by using mean-value theorem we have $|f_n(s,X_{s^-}+x)-f_n(s,X_{s^-})|=|\frac{\partial f_n}{\partial x}(s,C)|\,|x|$, where $C$ is a random variable between $X_{s^-}$ and $X_{s^-}+x$. By applying Lemma \ref{lemma:key} and the same procedure as Step 2, we can show that $|f_n(s,X_{s^-}+x)-f_n(s,X_{s^-})|\leq C^{'}|x|$, where $C^{'}$ is a finite random variable, free from $s$, $x$, $n$. On the other hand, since $X$ is a finite variation L\'evy process, we also have that $\integral{[0,t]\times\R}{}{|x|}J_X{(ds\times dx)}<\infty$, $\pr$-almost surely.

Therefore by applying Lebesgue Dominated Convergence theorem, one can interchange the limit and the integral in expression $I_n$ as $n$ goes to infinity. Since $L_{\tilde f}=\R\times U\supseteq[0,t]\times U$, and for all $s\geq0$, $X_s$ and $X_{s^-}$ are in $U$, by part three of Theorem \ref{theorem:convsmooth}, we get
\begin{align*}
   \lim_{n\rightarrow\infty}I_n &= \iint\limits_{[0,t]\times\R}\parenth{\tilde f(s,X_{s^-}+x)-\tilde f(s,X_{s^-})}\;J_X(ds\times dx)
   \\  &=\iint\limits_{[0,t]\times\R}\parenth{f(s,X_{s^-}+x)-f(s,X_{s^-})}\;J_X(ds\times dx).
\end{align*}

\textbf{Step 5.} From Steps 1, 2, 3, 4, for a fixed $t\geq0$, we have $\pr$-almost surely the following identity 
\begin{align}\label{eq:step5}
f(t,X_t) = f(0,X_0)&+\integral{0}{t}{\frac{\partial f}{\partial s}(s,X_s)}ds+\gamma\integral{0}{t}{\frac{\partial f}{\partial x}}ds\nonumber
\\  &+ \iint\limits_{[0,t]\times\R}\parenth{f(s,X_{s^-}+x)-f(s,X_{s^-})}\;J_X(ds\times dx).
\end{align}
The process $X$ is c\`adl\`ag, so the left-hand side and the right-hand side of the above equality are well defined processes. Therefore so far we have shown that the two sides of the above equation (when considered as processes) are in fact modifications of each other. Now we prove that as processes the left-hand side and the right-hand side are indeed indistinguishable.
\begin{enumerate}
   \item First note that since $f$ is continuous on $[0,\infty)\times U$, then $(f(t,X_t))_{t\geq0}$ is c\`adl\`ag.
   \item The function $\frac{\partial f}{\partial s}$ is Borel measurable and for a fixed $\omega\in\Omega$, $(X_s)_{0\leq s\leq t}$ is also Borel measurable. Hence for a fixed $\omega$, $\frac{\partial f}{\partial s}(s,X_s)$ is Borel measurable. So it is also Lebesgue measurable and by Fundamental theorem of Lebesgue integral calculus $t\longmapsto\integral{0}{t}{\frac{\partial f}{\partial s}(s,X_s)}ds$ is uniformly continuous in $t$. Note that in Step 2, we actually showed that $\frac{\partial\tilde f}{\partial s}(s,X_s)$ is Lebesgue integrable and on $[0,t]$, $\frac{\partial\tilde f}{\partial s}(s,X_s)=\frac{\partial f}{\partial s}(s,X_s)$.
   \item Similarly to the previous case, $t\longmapsto\integral{0}{t}{\frac{\partial f}{\partial x}(s,X_s)}ds$ is also continuous in $t$.
   \item Let $Z_t:=\iint\limits_{[0,t]\times\R}\parenth{f(s,X_{s^-}+x)-f(s,X_{s^-})}\;J_X(ds\times dx)$. For all $s\geq0$, $X_s$ and $X_{s^-}$ are in $U$, therefore  $Z_t=\sum_{0\leq s\leq t}\parenth{f(s,X_s)-f(s,X_{s^-})}$. If the function $f$ is $C^{1,1}$, then obviously the process $Z=(Z_t)_{t\geq0}$ is right continuous. However, since here $f$ is not necessarily smooth, to show the right continuity of $Z$, we do as follows:
\begin{align*}
   \lim_{h\rightarrow0^+}|Z_{t+h}-Z_t| & = \lim_{h\to0^+}|\sum_{t<s\leq t+h}\parenth{f(s,X_s)-f(s,X_{s^-})}|
\\  &\leq\lim_{h\rightarrow0^+}\sum_{t<s\leq t+h}|f(s,X_s)-f(s,X_{s^-})|\\
    &=\lim_{h\rightarrow0^+}\sum_{t<s\leq t+h}|\lim_{n\rightarrow\infty}\left(f_n(s,X_s)-f_n(s,X_{s^-})\right)|\\
    &\leq\lim_{h\rightarrow0^+}\sum_{t<s\leq t+h}C^{''}|\Delta X_s|,
\end{align*}
where similar to Step 4, one can show that $C^{''}$ is a finite random variable free from $s$, $h$, $n$, so we obtain $$\lim_{h\rightarrow0^+}|Z_{t+h}-Z_t|\leq C^{''}\lim_{h\rightarrow0^+}\sum_{t<s\leq t+h}\Delta X_s=0,\;\pr-\;\text{almost surely}.$$
This shows that the process $Z$ is right continuous.

 Thus the left-hand side and the right-hand side of equation \eqref{eq:step5}, when considered as processes, are right continuous, and we already know that they are also modification of each other. By Theorem 4, Chapter \Rmnum{1} of \cite{protter}, we conclude that the left-hand side and the right-hand side of this equation define two processes that are indistinguishable.

This proves our theorem.
\end{enumerate}
\end{proof}
The next example shows that even in one dimensional cases, there are simple functions for which Meyer-It\^o formula is not applicable but Theorem \ref{theorem:main} can be used.
\begin{example}
     Assume that $X:[0,\infty)\times\Omega\longmapsto\R$ is a finite variation L\'evy process that satisfies Assumption   \ref{assump:main}. Let the function $f:\R\longmapsto\R$ be defined by 
$$ f(x)=\left\{
         \begin{array}{ll}
           x^2\sin(\frac{1}{x}), & \hbox{$x\neq0$;} \\
           0, & \hbox{$x=0.$}
         \end{array}
\right.$$
This function is continuous, but its derivative is not continuous at origin. So the classical It\^o's formula cannot be applied. Moreover, one can show that $f$ cannot be written as the difference of two convex functions, and hence Meyer-It\^o's formula (Theorem 70, Chapter \Rmnum{4} of \cite{protter}) is not applicable as well. However, $f$ is weakly differentiable, its weak derivative is locally bounded, and therefore Theorem \ref{theorem:main} is in force.
\end{example}
\begin{example}
     Let the function $f$ and the process $X$ be the same as Theorem \ref{theorem:main}. In addition, we equip the probability space $(\Omega,\mathfrak F,\pr)$ with the natural filtration $\mathbb F^X=\{\mathcal F_t; t\geq0\}$ generated by the history of $X$, i.e. for each $t\geq0$, $\mathcal F_t$ is the sigma algebra generated by $\{X_s; s\leq t\}$ and all the null sets of $\mathfrak F$.  Since $X$ is a finite variation L\'evy process, similar to Step 4 of Theorem \ref{theorem:main}, one can show that for every $t\geq0$, $\iint\limits_{[0,t]\times\R}\left|f(s,X_{s^-}+x)-f(s,X_{s^-})\right|\;ds\times v(dx)<C\int_\R x\;v(dx)<\infty,$ $\pr$-almost surely, where $C$ is a random variable free from $s$ and $x$.  Then we have the following decomposition: $f(t,X_t)=f(0,X_0)+M_t+\int_0^t\mathcal A f(s,X_s)\;ds$, where $M$ is a local martingale with respect to $\mathbb F^X$ given by $M_t=\iint\limits_{[0,t]\times\R}\parenth{f(s,X_{s^-}+x)-f(s,X_{s^-})}\;\tilde J_X(ds\times dx)$, $\tilde J_X(ds\times dx)=J_X(ds\times dx)-ds\times v(dx)$, and 
     
     \begin{align*}
\mathcal A f(s,X_s) = \frac{\partial f}{\partial s}(s,X_s)+\gamma\frac{\partial f}{\partial x}(s,X_s)
+ \int_\R\parenth{f(s,X_{s^-}+x)-f(s,X_{s^-})}\;v(dx).
\end{align*}
In other words, this shows that the process $(f(t,X_t))_{t\geq0}$ is a special semimartingale.
\end{example}
In the next lemma, we get back to the motivation provided in the introduction. This lemma also highlights applications of Theorem \ref{theorem:main} in Feynman-Kac representations. Comparing to similar results, for instance \cite{rong}, this representation is valid in the absence of diffusions terms. In addition, there are less restrictive assumptions on the underlying function.
\begin{lemma}
Suppose that $X$ is a finite variation L\'evy process that satisfies Assumption \ref{assump:main} for $U=\R$. Let the function $P=P(t,x)$, defined by equation \eqref{eq:purposed}, admit $L_{loc}^1([0,T]\times(0,\infty))$-weak derivatives which are locally bounded. Then using Theorem \ref{theorem:main} and following the same procedure as Proposition 12.2 of \citet{cont} (or \cite{rong}), one can show that $P=P(t,x)$ is the solution of PIDE \eqref{eq:pide}.
\end{lemma}

\section{Conclusions}
A version of It\^o's formula is studied under multi-dimensional finite variation L\'evy processes that is time-dependent and requires weak differentiability. The formula can be particularly useful for functions that are continuous and piecewise smooth. The possible formula's applications were motivated by a financial example. 

The two main assumptions are that the process is finite variation and the weak derivatives of the functions are locally bounded. The extension of the formula to pure jump semimartingales, using the theory of distributions (in functional analysis),  is interesting for future work.

\section{Acknowledgments}
The first author gratefully acknowledges partial financial support from the Vienna Science and Technology Fund (WWTF) under grant MA09-005 at Vienna University of Technology. Also, the authors are thankful to an anonymous referee for his/her constructive
comments.


\bibliographystyle{plainnat}

\end{document}